\documentclass[conference,letterpaper,onecolumn]{IEEEtran}
\usepackage[doublespacing,nodisplayskipstretch]{setspace}
\usepackage{balance}
\usepackage[utf8]{inputenc} 
\usepackage[T1]{fontenc}
\usepackage{url}
\usepackage{ifthen}
\usepackage{cite,enumitem}
\usepackage{mathrsfs}  
\usepackage[cmex10]{amsmath} 
\usepackage{graphicx,mathtools}
\usepackage{amssymb,amsthm,verbatim}
\usepackage[colorlinks=true,citecolor=blue]{hyperref}
\usepackage{tcolorbox} 
\tcbuselibrary{listingsutf8} 

\newcommand{\hnew}[1]{\textcolor{red}{#1}}

\definecolor{lightgray}{RGB}{224,224,224}

\newtheorem{theorem}{Theorem}

\newtheorem{exmp}{Example}
\newtheorem{corollary}{Corollary}
\newtheorem{definition}{Definition}
\newtheorem{remark}{Remark}
\newtheorem{proposition}{Proposition}
\newtheorem{lemma}{Lemma}

\begin{document}
\title{Majority-Logic Decoding of Binary Locally Recoverable Codes: A Probabilistic Analysis}

\author{%
  \IEEEauthorblockN{Hoang Ly, Emina Soljanin}
  \IEEEauthorblockA{
  Rutgers University, USA}
                     \{\texttt{mh.ly;emina.soljanin}\}@rutgers.edu
  \and
  \IEEEauthorblockN{Philip Whiting}
  \IEEEauthorblockA{
Macquarie University, Australia }
                     philipawhiting@gmail.com
}
\maketitle
\begin{abstract}
   Locally repairable codes (LRCs) were originally introduced to enable efficient recovery from erasures in distributed storage systems by accessing only a small number of other symbols. While their structural properties—such as bounds and constructions—have been extensively studied, the performance of LRCs under random erasures and errors has remained largely unexplored. In this work, we study the error- and erasure-correction performance of binary linear LRCs under majority-logic decoding (MLD). Focusing on LRCs with fixed locality and varying availability, we derive explicit upper bounds on the probability of decoding failure over the memoryless Binary Erasure Channel (BEC) and Binary Symmetric Channel (BSC). Our analysis characterizes the behavior of the bit-error rate (BER) and block-error rate (BLER) as functions of the locality and availability parameters. We show that, under mild growth conditions on the availability, the block decoding failure probability vanishes asymptotically, and that majority-logic decoding can successfully correct virtually all of error and erasure patterns of weight linear in the blocklength. The results reveal a substantial gap between worst-case guarantees and typical performance under stochastic channel models.
\end{abstract}

\section{Introduction}
Locally repairable codes (LRCs) are a class of error-correcting codes designed to enable the recovery of individual code symbols by accessing only a small subset of other symbols, known as \emph{local recovery sets} (also referred to as repair groups). This property, formalized through the notions of \emph{locality} and \emph{availability}, makes LRCs particularly attractive for distributed storage systems, where node failures lead to data loss modeled as erasures~\cite{Gopalan}. By providing multiple, disjoint recovery sets for each symbol, LRCs allow local repair with low bandwidth and latency overheads, and have therefore received significant attention in recent years, with applications extending beyond distributed storage~\cite{PIR_LRC_equivalent}.

Although LRCs were originally introduced to combat erasures, the presence of multiple recovery sets also suggests potential robustness to \emph{errors}. In practical storage and memory systems, failures are often better modeled as random errors---such as bit flips---rather than erasures, reflecting data corruption at the level of memory elements~\cite{MLD_LDPC}. Such bit flips may arise from unavoidable physical effects and remain a concern even in highly reliable systems, including safety-critical applications~\cite{cosmic_rays_aircraft_2025}. When errors are present, corrupted symbols may invalidate some recovery sets, thereby degrading the reliability of local recovery. This observation motivates the use of \emph{majority-logic decoding} (MLD), a classical decoding paradigm first introduced for Reed--Muller (RM) codes~\cite{reed1954class,Massey1963Threshold}. Under MLD, each recovery set produces a local estimate (a \emph{vote}) for the target symbol, and the decoder outputs the majority of these votes (with an arbitrary tie-breaking rule when needed)~\cite{One_step_RM}. In the binary case, each local estimate is a parity computation and can be implemented using \textsf{XOR} operations. \textcolor{black}{More broadly, highly efficient erasure-coding families such as Tornado codes and fountain codes (e.g., LT and Raptor codes) rely on sparse parity constraints and iterative decoding procedures whose computations are mostly just \textsf{XOR} operations, enabling low-complexity (often linear- or near-linear-time) encoding and decoding on erasure channels~\cite{Tornado_codes,LT_code,Raptor_code}.}

Majority-logic decoding is a simple, hard-decision decoding method with low complexity and extremely small decoding delay, which historically made it attractive for applications such as deep-space communication and memory protection \cite{Class_Of_MLDcodes,MLD_memory,OS_MLD_memory_protection}. Despite major advances in decoding algorithms and hardware capabilities, there has been a renewed interest in such low-complexity decoders for latency-critical applications, where hard-decision algorithms are often preferred over soft-decision alternatives \cite{RMdecoding:journals/tcom/BertramHH13}. Recently, MLD has also been connected to combinatorial designs, and to the concept of \emph{Service Rate Region}, a metric that quantifies the efficiency of linear codes for optimizing data access~\cite{SRR_Design:lySV2025}. Moreover, MLD is not a bounded-distance decoder and is known to frequently correct more errors than other bounded-distance decoders \cite{RMdecoding:journals/tit/Chen71}. For instance, over long binary RM codes of fixed rate, MLD can correct most error patterns of weight up to $d\log d/4$, where $d$ denotes the minimum distance~\cite{Krichevskiy1970RM}. For certain code families, such as Grassmann codes, MLD remains the only known method enabling efficient and practical decoding \cite{Grassmannian_code_MJ}. Codes that admit such decoding algorithms are referred to as \emph{majority-logic decodable codes}~\cite{Coding:books/PetersonW72}.

Surprisingly, despite the natural compatibility between LRCs and MLD, existing works have almost exclusively focused on the use of LRCs for erasure recovery. Even when robustness to errors is considered—for example, allowing recovery of an erased symbol in the presence of errors in parity-check symbols~\cite{local-error-correction}—the explicit use of LRCs to correct random errors has not been systematically studied. In particular, no prior work has analyzed the performance of LRCs under standard probabilistic error models, such as the Binary Symmetric Channel (BSC), using MLD. Most of the existing literature on LRCs concentrates on code constructions and fundamental limits, including bounds on minimum distance, optimal length, symbol size, and channel capacity~\cite{Capacity_LRC}. From an MLD perspective, the availability property of LRCs—namely, the presence of multiple recovery sets per symbol—naturally suggests increased robustness, as it enables majority voting across independent recovery equations and thus greater tolerance to erasures and errors. However, despite this intuition, the efficiency of LRCs against random erasures and errors—quantified in terms of decoding failure probability, bit-error rate (BER), or block-error rate (BLER)—has remained largely unexplored. In contrast, for classical code families such as Reed--Muller codes, where majority-logic decoding can be used efficiently, such probabilistic performance analyses were established quite some time ago~\cite{Krichevskiy1970RM,soft_decision_Dumer}. While Reed--Muller codes can be viewed as early examples of codes with local recovery structure, they are not typically designed or employed for modern distributed storage systems.

\textcolor{black}{We note that LRCs constitute a broadly defined class of codes characterized primarily by their local recovery properties—namely, locality and, in some cases, availability—rather than by a specific algebraic construction such as a generator or parity-check matrix. This distinguishes LRCs from classical code families, for which decoding algorithms are often tailored to explicit algebraic structures (e.g., the Berlekamp--Welch algorithm for Reed--Solomon codes or the Berlekamp--Massey decoder for BCH codes). As a result, there is currently no canonical, structure-specific decoding algorithm for general LRCs. In this sense, MLD emerges as a natural and broadly applicable decoding framework for LRCs, as it relies solely on the presence of multiple local recovery sets inherent to the codes' definition.}

\subsection*{Contributions}
\begin{itemize}
    \item \textbf{MLD for error correction in LRCs:}
We extend the use of binary linear locally recoverable codes—originally designed for erasure recovery in distributed storage—to error correction via majority-logic decoding (MLD), a simple, practical, and low-complexity hard-decision decoder originally developed for Reed–Muller codes. In contrast to the joint-typicality decoding approach presented in~\cite{Capacity_LRC}, MLD does not require knowledge of the channel law and is well suited for latency- and complexity-constrained settings.

\item \textbf{Probabilistic performance analysis under BSC and BEC:}
We initiate a systematic probabilistic analysis of the error- and erasure-correction performance of LRCs under MLD over the memoryless Binary Symmetric Channel (BSC) and Binary Erasure Channel (BEC). We derive explicit bounds on the probability of decoding failure as functions of the locality $r$ and availability $t$.

\item \textbf{Typical-case decoding guarantees beyond worst-case bounds:}
Moving beyond deterministic, worst-case guarantees based on minimum distance or availability, we characterize the typical decoding behavior of LRCs under random errors and erasures. Our results quantify how many errors and erasures can be corrected with high probability using MLD and reveal a substantial gap between worst-case adversarial limits and stochastic performance.

\item \textbf{Generalization of classical Reed–Muller results:}
Our analysis generalizes classical probabilistic decoding results for Reed–Muller codes~\cite{Krichevskiy1970RM} as special cases, while addressing a previously unexplored aspect of LRC performance under MLD.
\end{itemize}
This paper is organized as follows. Section~\ref{sec:LRC_definition} introduces binary LRCs with availability and reviews the MLD framework. Section~\ref{sec:analysis} presents a probabilistic analysis of the decoding performance of LRCs under MLD over the BSC and BEC, including bounds on BER and BLER, and weight of correctable error and erasure patterns in the asymptotic regimes of block length $n$. Section~\ref{sec:simulation} provides simulation results that validate the theoretical analysis and illustrate the impact of availability scaling. Section~\ref{sec:conclusion} concludes the paper and discusses directions for future work.

In the sequel, \(\log\) denotes the base-2 logarithm, $\ln$ denotes the natural logarithm, and \(\exp(x)\) denotes \(2^{x}\) while $\mathrm{e}$ denotes the Euler's constant. The binary field is denoted as \(\mathbb{F}_2\), and binary linear code \(\mathcal{C}\) with parameters \([n, k, d]_2\) is a \(k\)-dimensional subspace of the \(n\)-dimensional vector space \(\mathbb{F}_2^n\) with minimum Hamming distance $d$. Set of positive integers not exceeding \(i\), namely $\{1, 2,\dots, i\}$, is denoted as \( \lbrack i \rbrack \).

\section{Locally Recoverable Codes with Availability}
\label{sec:LRC_definition}
In this section, we will introduce the Locally Recoverable Codes with Availability, and its decoding method, namely Majority-logic decoding.
\subsection{Locally Recoverable Codes with Availability}
We follow the conventional definitions of binary linear LRCs with
availability established in, for instance,~\cite{Binary_LRC,Bounds_LRC}.
\begin{definition}
\label{defin:LRC}
The $i$-th code symbol of an $[n, k, d]_2$ binary linear code \(\mathscr{C} \subseteq \mathbb{F}_2^n\) is said to have locality $r$ and availability $t$ if, for every codeword $\mathbf{c} = (c_1, c_2, \ldots, c_n) \in \mathscr{C}$, there exist $t$ pairwise disjoint recovery sets (also referred to as repair sets) $R^{1,i},\, R^{2,i},\, \ldots,\, R^{t,i} \subseteq [n] \setminus \{i\}$ such that:
\begin{enumerate}
    \item Each set has cardinality at most $r$, i.e., $|R^{j,i}| \le r, \ \forall\, j \in [t]$, and
    \item For each recovery set $R^{j, i}$, the symbol $c_i$ can be recovered via the sum of the symbols in that set, i.e., $c_i = \sum\limits_{u \in R^{j,i}}c_u, \, \forall\, j \in [t]$, where the sum is over $\mathbb{F}_2$.
\end{enumerate}
\end{definition}

\begin{definition}
    The linear code $\mathscr{C}$ is said to have \emph{all-symbol} locality $r$ and availability $t$ if every code symbol $i \in [n]$ has locality $r$ and availability $t$. We refer to such a code $\mathscr{C}$ as a linear $(r, t)_a$-LRC.
\end{definition}
The properties of LRCs imply that every codeword symbol can be recovered from \(t\) pairwise disjoint local recovery sets, or equivalently, from \(t\) independent \emph{votes}. This structural recovery naturally motivates us to use \emph{majority-logic decoding (MLD)}, which was first introduced by Reed in 1954 for Reed--Muller codes~\cite{reed1954class,Massey1963Threshold}, where the final decision on each symbol is determined by the majority of its local votes. Codes that can be decoded using MLD are called majority-logic decodable codes. Recently, authors in~\cite{Binary_LRC} presented a construction of a class of long binary LRC with availability based on majority-logic decodable codes. 

In this write-up, we focus on a class of binary linear LRCs \(\mathscr{C} \subseteq \mathbb{F}_2^n\) in which: (1) all recovery sets have the same size $r$, where $r$ is fixed, and (2) the number of disjoint recovery sets $t$ grows with the code length $n$ (as in the case of Simplex codes). We justify that the assumption of constant recovery set size (1) is both valid and reasonable, as used commonly in the literature, for example, in~\cite{Bounds_LRC,LRC_download_time_analysis}. Indeed, this assumption not only helps simplify our analysis but is practically solid. If a recovery set of size $s < r$ exists, it can always be formally treated as a set of size $r$ by add into it other $r-s$ available symbols among the remaining. Since these extra objects do not participate in the actual recovery process, the (pairwise) disjointness of the recovery sets is preserved. Moreover, in hardware/VLSI (Very-Large-Scale Integration) systems, it is significantly more efficient to design, implement, and maintain circuits with a uniform number of inputs rather than supporting varying input sizes~\cite{harris2013digital}, especially as the number of required circuits scales. Such codes are often referred to as \emph{LRCs with large availability} and have attracted significant practical interest~\cite{LRC_large_availability}. Explicit constructions of binary linear LRCs with fixed locality and availability growing with $n$ have been presented, for example, in~\cite{LRC_large_availability} (linear growth) and~\cite{LRC_availablity_construction,Bounds_LRC}. Authors in~\cite{LRC_scaling_locality} presented a construction of a class of LRC with scaling locality and availability that required large field sizes. Authors in~\cite{arbitrary_locality_availability} presented a construction of a class of binary code that achieves arbitrary locality and availability. Interested readers are referred to~\cite{Codes_distributed_storage} for a comprehensive and detailed review.
\subsection{Majority-Logic Decoding}
Suppose the original codeword \(\mathbf{c}\) is transmitted over a communication channel and the received (possibly corrupted) sequence is denoted by $\mathbf{u} = \mathbf{c} + \mathbf{e}$, where \(\mathbf{u} = (u_1, u_2, \ldots, u_n) \in \mathbb{F}_2^n\), and $\mathbf{e}\in \mathbb{F}_2^{n}$ is a binary noise vector. For a fixed \(i\) and let \(R^{j,i}\) be one of the recovery sets for symbol $c_i$, the decoder forms a local estimate \(\hat{c}_i^{(j)}\) of \(c_i\) using
\[
\hat{c}_i^{(j)} = \sum_{h \in R^{j,i}} u_h, \quad \forall \, j \in [t].
\]
In practice, these summations over $\mathbb{F}_2$ can be implemented efficiently using $\textsf{XOR}$ gates. Finally, a global decision on \(c_i\) is made by taking the majority vote among the local estimates:
\[
\hat{c}_i =
\begin{cases}
1, & \text{if a majority (i.e., at least half) of } \hat{c}_i^{(j)} = 1, \\[4pt]
0, & \text{otherwise.}
\end{cases}
\]
\begin{exmp}
Consider a $[7,\,3]_2$ Simplex code whose length-$7$ codewords
$\mathbf{c}=(c_1,c_2,c_3,c_4,c_5,c_6,c_7)$ lie in the row space of the generator matrix
\[
\mathbf{G}=
\begin{bmatrix}
0 & 0 & 0 & 1 & 1 & 1 & 1\\
0 & 1 & 1 & 0 & 0 & 1 & 1\\
1 & 0 & 1 & 0 & 1 & 0 & 1
\end{bmatrix}.
\]
In this code, the symbol $c_1$ admits three disjoint recovery sets and can be written as
\begin{equation}
\label{eq:hamming_votes}
c_1 = c_2 + c_3 = c_4 + c_5 = c_6 + c_7 \quad (\text{over } \mathbb{F}_2).
\end{equation}
Equivalently, $R^{1,1} = \{2,\, 3\},\, R^{2,1} = \{4,\, 5\},\, R^{3,1} = \{6,\, 7\}$.
Suppose the received vector $\mathbf{u}$ differs from the transmitted codeword $\mathbf{c}$ due to channel errors. Since the recovery sets in~\eqref{eq:hamming_votes} are pairwise disjoint, each bit error affects at most one recovery set. For example, if only $u_2$ is flipped while all other symbols are received correctly, then the estimate $u_2+u_3$ disagrees with $u_4+u_5$ and $u_6+u_7$. A majority vote among
\[
u_2+u_3,\quad u_4+u_5,\quad u_6+u_7
\]
correctly recovers $c_1$. Moreover, even if two bit errors occur within the same recovery set—for instance, $u_2=c_2+1$ and $u_3=c_3+1$—the corresponding estimate remains correct:
\[
u_2+u_3=(c_2+1)+(c_3+1)=c_2+c_3 \quad (\text{over } \mathbb{F}_2).
\]
\color{black} In contrast, an error pattern of the same weight affecting two distinct recovery sets (e.g., errors at $u_2$ and $u_4$, i.e., $\mathbf{e}=0101000$) causes the majority vote to output an incorrect estimate of $c_1$. This highlights that a \emph{given} recovery set produces an incorrect vote if and only if it contains an odd number of errors.

Such an error pattern contributes to the bit error rate (BER) because the decoded value of $c_1$ is incorrect, and it also contributes to the block error rate (BLER) under symbol-wise decoding, since a block is declared correct only when all symbols are decoded correctly.

\end{exmp}

In general, a local recovery set containing an \emph{even} number of errors yields a correct vote, whereas a recovery set containing an \emph{odd} number of errors yields an incorrect vote. Thus, the correctness of each vote depends only on the parity of the number of errors within the recovery set, not on the total number of errors. Consequently, majority-logic decoding can successfully correct certain error patterns that conventional bounded-distance decoding methods cannot, despite relying solely on local parity constraints. This observation motivates our probabilistic analysis of majority-logic decoding for $(r,t)_a$ LRCs over channels with independent bit flips (also called crossovers) and erasures, in terms of BER, BLER, and the number of correctable errors and erasures. \textcolor{black}{On a related note, another example of codes where error correction capability depends on error pattern and not on the number of errors itself is Single parity-check Product codes~\cite{product_code}.}

\textcolor{black}{Following~\cite{RM_recursive_Dumer}, to define their error-correcting performance, we use the following definition. Given an infinite sequence of codes $A_i(n_i, d_i)$, we say that a decoding algorithm $\Psi$ has a \emph{threshold} sequence $\delta_i$ and a \emph{residual} sequence $\epsilon_i$ if for $n_i \rightarrow \infty$, $\Psi$ correctly decodes all but a vanishing fraction of error
patterns of weight $\delta_i(1-\epsilon_i)$ or less.}


\section{Analysis of Decoding Capability of LRCs}
\label{sec:analysis}
In this section, we present two complementary approaches for analyzing the error- and erasure-correction capability of LRCs under majority-logic decoding (MLD): a deterministic approach and a probabilistic approach. The former characterizes worst-case (adversarial) correction guarantees, while the latter focuses on typical performance when errors and erasures occur randomly and independently. Our main results are derived under this probabilistic channel model.

\subsection{Deterministic versus Probabilistic Perspectives}
\paragraph{Deterministic approach.}
Typically, the correction capability of an LRC is characterized via minimum distance or availability under an adversarial (i.e., worst case) error model. In the worst case, even if \(t-1\) recovery sets are completely lost (for example, each containing at least one erasure), the symbol can still be recovered using the remaining intact set. Hence, deterministically, an \((r,t)_a\)-LRC can tolerate up to \(t-1\) erasures, irrespective of the channel statistics~\cite{multiple_erasures}. This represents the \emph{worst-case} or \emph{adversarial} erasure-correction capability of the code. Similarly, the code can tolerate at most \(\lfloor (t-1)/2 \rfloor\) errors, since in this case the majority of local parity votes remain correct, allowing MLD to successfully recover the target symbol.

However, such worst-case bounds are often pessimistic and not particularly informative in practice. \textcolor{black}{In many systems, bursty or dependent error patterns can be mitigated by \emph{interleaving}, which spreads error bursts across code symbols and can make the resulting error process appear approximately independent at the decoder~\cite{interleaver}.} In contrast, a more common and realistic scenario is one in which errors and erasures occur randomly~\cite{Capacity_LRC}. Under such stochastic behavior, errors often affect multiple symbols within the same recovery set, leaving other recovery sets completely intact. This sharply contrasts with the adversarial scenario, where errors are distributed to corrupt as many recovery sets as possible. Consequently, the number of corrupted recovery sets is often significantly smaller than the total number of raw errors, allowing successful decoding even beyond worst-case guarantees. This explains why MLD frequently corrects more than \(\lfloor (t-1)/2 \rfloor\) errors in practice and motivates the probabilistic analysis that follows.

\paragraph{Probabilistic approach.}
A more realistic view assumes that erasures or bit-flips occur independently according to a stochastic channel law, such as the memoryless $\mathrm{BEC}$ or $\mathrm{BSC}$. Under this model, we analyze the \emph{probability of decoding failure} at both the symbol (BER) and block (BLER) levels as functions of the channel parameters (bit-flip probability \(p_f\) or erasure probability \(p_e\)) and the code parameters \((n,r,t)\).

\subsection{Binary Symmetric Channel ($\mathrm{BSC}$)}
In this case, each bit is independently flipped with probability $p_f$ (``data corruption'') where $0 \le p_f < 0.5$. We now characterize the symbol/bit decoding failure probability of MLD in this setting.

\begin{theorem}
\label{thm:error_prob}
Consider transmission over a $\mathrm{BSC}(p_f)$ with $\mathrm{i.i.d.}$ errors, where $0 \le p_f < 0.5$ is a fixed bit-flip probability. For binary linear \((r, t)_a\)-LRC, 
the symbol (or bit) decoding failure probability of MLD 
satisfies
\begin{align*}
P^{\mathrm{BSC}}_{\mathrm{fail,bit}}
& \le \bigl(2\sqrt{a(1-a)}\bigr)^t = \big(1-(1-2p_f)^{2r}\big)^{t/2}, \quad\text{where}\quad
a = \frac{1 + (1 - 2p_f)^r}{2}.
\end{align*}
\end{theorem}

\begin{proof}
For a symbol $x$ and any of its associated recovery sets $R^x$ of size $r$, let
\begin{align*}
W & \triangleq \Pr[R^x\text{ gives a wrong vote}]\\
& = \Pr[R^x\text{ contains an odd number of bit flips}] \\
& = \sum_{\text{odd } s}\binom{r}{s}p_f^s(1-p_f)^{r-s},
\end{align*}
and
\begin{align*}
C & \triangleq \Pr[R^x\text{ gives a correct vote}] \\
& = \Pr[R^x\text{ contains an even number of bit flips}] \\
& = \sum_{\text{even } s}\binom{r}{s}p_f^s(1-p_f)^{r-s}.
\end{align*}
Clearly, $W+C=1$, and by the Binomial theorem expansion,
\begin{align*}
C-W & = \sum_{\text{even }s}\binom{r}{s}p_f^s(1-p_f)^{r-s} - \sum_{\text{odd } s}\binom{r} {s}p_f^s(1-p_f)^{r-s} \\
& = ((1-p_f)-p_f)^r = (1-2p_f)^r.  
\end{align*}
Hence,
\begin{align}
\label{eq:recovery_set_probability}
C = \frac{1 + (1 - 2p_f)^r}{2}, 
\qquad
W = \frac{1 - (1 - 2p_f)^r}{2}.
\end{align}
Note that since $0 \le p_f < 0.5$, we have $W < \frac{1}{2} < C$.

\emph{Random-variable formulation.}
For each recovery set $R^x$, define
\[
\xi_x =
\begin{cases}
+1, & \text{if $R^x$ gives a correct vote} \\
-1, & \text{if $R^x$ gives a wrong vote.}\\
\end{cases}
\]
Then $\Pr[\xi_x = 1] = C$ and $\Pr[\xi_x = -1] = W$. 
MLD fails when the total number of wrong votes exceeds the correct ones, i.e.,
\[
\sum_{x=1}^{t} \xi_x \le 0.
\]
Note that different $\xi_x$ are obtained from pairwise disjoint recovery sets $R^x$ and therefore \emph{independent}. Thus, the decoding error probability is
\(
P^{\mathrm{BSC}}_{\mathrm{fail,bit}}
= \Pr\!\left(\sum_{x=1}^{t} \xi_x \le 0\right).
\)
The result follows immediately using Lemma~\ref{lem:Chernoff}.
\end{proof}

\begin{lemma}[Chernoff bound]
\label{lem:Chernoff}
Let $X_1,\dots,X_t$ be $\mathrm{i.i.d.}$ random variables taking values in $\{-1, 1\}$, with $\Pr(X_i = 1) = a$ and $\Pr(X_i = -1) = 1-a$. If $a > 1/2$, then for any $t \ge 1$, the cumulative sum $S_t = \sum_{i=1}^t X_i$ satisfies
\[
\Pr(S_t \le 0) \le 
{\left(2\sqrt{a(1-a)}\right)^t}.
\]
\end{lemma}

\begin{proof}
The result follows from the standard \emph{exponential moment method} (see, e.g.,~\cite[Ch.~2]{concentration}). For any $\lambda < 0$, Markov's inequality applied to the random variable $e^{\lambda S_t}$ implies
\[
\Pr(S_t \le 0) = \Pr(e^{\lambda S_t} \ge 1) \le \mathbb{E}[e^{\lambda S_t}].
\]
By independence, the moment generating function factorizes as $\mathbb{E}[e^{\lambda S_t}] = \prod_{i=1}^t \mathbb{E}[e^{\lambda X_i}] = (a e^\lambda + (1-a)e^{-\lambda})^t$. We minimize the base term $\phi(\lambda) = a e^\lambda + (1-a)e^{-\lambda}$ with respect to $\lambda$. Setting the derivative to zero yields the optimal choice $e^{\lambda} = \sqrt{(1-a)/a}$. Since $a > 1/2$, we have $e^\lambda < 1$, ensuring $\lambda < 0$ as required. Substituting this value back into $\phi(\lambda)$ yields the bound $\big(2\sqrt{a(1-a)}\big)^t$.
\end{proof}

\begin{remark}[Exponential Tightness]
\label{rm:tightness}
The bound is exponentially tight; that is, the base of the exponential term cannot be improved. This follows from Cramér's theorem in the theory of Large Deviations (see~\cite[Ch.~2]{dembo2009large}). The rate of decay is governed by the Kullback–Leibler divergence between a Bernoulli$(1/2)$ and a Bernoulli$(a)$ distribution:
\begin{align*}
& \lim_{t \to \infty} -\frac{1}{t} \ln \Pr(S_t \le 0) = D_{\mathrm{KL}}(1/2 \,\|\, a) \\
& = \frac{1}{2} \ln\left(\frac{1/2}{a}\right) + \left(1-\frac{1}{2}\right) \ln\left(\frac{1-1/2}{1-a}\right) \\
& = -\ln\left(2\sqrt{a(1-a)}\right).
\end{align*}
Consequently, $\Pr(S_t \le 0) \approx \mathrm{e}^{ -t \cdot D_{\mathrm{KL}}(1/2 \,\|\, a)} = \left(2\sqrt{a(1-a)}\right)^t$. This confirms that both the base and the linear dependence on $t$ are asymptotically optimal. 
\end{remark}



\color{black}{Theorem~\ref{thm:error_prob} makes no assumption on how the availability \(t\) scales with the blocklength \(n\). In particular, the bound holds uniformly for all \(n\) and \(t\), whether \(t\) is fixed, grows slowly, or grows unbounded with \(n\), and it implies an exponential decay (in \(t\)) of the \emph{upper bound} on the bit decoding failure probability. On the other hand, it follows from~\eqref{eq:recovery_set_probability} that the probability \(C\) that a given recovery set produces a correct vote is a decreasing function of the locality \(r\) for any fixed bit-flip probability \(p_f \in [0,0.5)\); equivalently, smaller recovery sets are more likely to yield a correct vote.

Taken together, these observations reveal a basic \emph{contrast}. From the perspective of a target symbol, having more recovery sets (larger \(t\)) provides more independent votes and improves reliability. From the perspective of each recovery equation (parity constraint), smaller locality (smaller \(r\)) improves the reliability of the individual vote. Interpreting the recovery structure through a bipartite Tanner-like graph, these preferences resemble the tension between variable-node connectivity and check-node degree observed in iterative decoding of LDPC codes (e.g., under belief propagation)~\cite{irregular_LDPC}.}

Theorem~\ref{thm:error_prob} leads to the following corollary on the block decoding failure probability (BLER).
\begin{theorem}[\textbf{MLD Success for LRCs on the BSC}]
\label{thm:block_flip_prob}
Consider a binary \((r,t)_a\)-LRC where $r$ is fixed and $t = t(n)$ grows with $n$. When transmitted over a $\mathrm{BSC}(p_f)$ with fixed error probability $0 \le p_f < 0.5$, the code achieves successful block decoding asymptotically if the availability scales sufficiently fast. Concretely, if $t(n) = \omega(\log n)$ (i.e., $t(n)$ grows faster with $n$ than $\log n$ or $\lim\limits_{n\to\infty} \frac{t(n)}{\log n} = \infty$), then the probability of block decoding failure vanishes in the asymptotic regime of $n$:
\[
    \lim\limits_{n\rightarrow \infty} P_{\mathrm{fail,block}}^{\mathrm{BSC}} = 0.
\]
\end{theorem}
\begin{proof}
For any fixed symbol, Theorem~\ref{thm:error_prob} gives the bit-failure probability
\[
P_{\mathrm{fail,bit}}^{\mathrm{BSC}}
    \le \bigl(1 - (1-2p_f)^{2r}\bigr)^{t/2}.
\]
Using the union bound, the failure probability of decoding blocks (or words) satisfies
\[
P_{\mathrm{fail,block}}^\mathrm{BSC}
    \;\le\; n \, P_{\mathrm{fail,bit}}^\mathrm{BSC}
    \le n \bigl(1 - (1-2p_f)^{2r}\bigr)^{t/2} .
\]
Define
\[
x \triangleq (1-2p_f)^{2r} \in (0,1],
\]
so that the upper bound becomes \(n(1-x)^{t/2}\).  
Using the standard inequality \(\log(1-x) < -x\) for all \(x \in (0,1)\), we may write
\[
\log(1-x) = -x - \varepsilon,
\qquad \text{for some } \varepsilon > 0.
\]
Thus
\[
n(1-x)^{t/2}
    = \exp\!\left(\log n - \frac{t}{2}x - \frac{t}{2}\varepsilon\right).
\]

Since \(x = (1-2p_f)^{2r}\) and $\varepsilon$ are both positive, the exponent tends to \(-\infty\) provided
\[
\lim_{n\to \infty}\frac{t(n)}{\log n} \;\longrightarrow\; \infty.
\]
In this case we have \(n(1-x)^{t/2} \to 0\), and therefore  
\(P_{\mathrm{fail,block}}^{\mathrm{BSC}} \to 0\). This completes the proof.
\end{proof}

The result above leads to the following interesting result.

\begin{theorem}[\textbf{MLD of LRCs with availability $t$ and locality $r$ over the $\mathrm{BSC}$}]
\label{thm:error_weight}
Let \(c \ge 3\) be any constant. For a binary \((r,t)_a\)-LRC with availability \(t = t(n)\) growing with $n$, all but a vanishing fraction of error patterns of weight
\begin{equation}
\label{eq:error_weight}
w \;\le\; \frac{n}{2}\!\left(1 - \sqrt[2r]{\frac{c \log n}{t(n)}}\right)    
\end{equation}
are correctly recovered by majority-logic decoding as \(n \to \infty\). In other words, virtually all error patterns of weight $w$ are correctable in the asymptotic regime of $n$. \textcolor{black}{In other words, long RM codes of growing availability can be decoded with complexity $O(nt(n))$ and decoding threshold $\delta = \dfrac{n}{2},\, \epsilon = \sqrt[2r]{\frac{c \log n}{t(n)}}$.}
\end{theorem}

\begin{proof}
Consider a $\mathrm{BSC}(p)$ channel in which
\begin{align}
p \triangleq \frac{1}{2}\!\left(1 - \sqrt[2r]{\frac{c\log n}{t}}\right),
\qquad
x \triangleq (1-2p)^{2r} = \frac{c\log n}{t}, 
\label{eq:setting_p}
\end{align}
so that error patterns of weight \(w = pn\) correspond to the threshold in the right-hand side of~\eqref{eq:error_weight}. Using Union bound and Theorem~\ref{thm:error_prob}, the block decoding failure probability under majority-logic decoding satisfies
\[
P_{\mathrm{fail,block}}^{\mathrm{BSC}} \le nP_{\mathrm{fail,bit}}^{\mathrm{BSC}}
    \;\le\; n (1-x)^{t/2}.
\]
Using \(\log(1-x) = -x - \varepsilon\) for some \(\varepsilon>0\) and \(x\in(0,1)\),
\begin{align}
\label{eq:fail_block}
P_{\mathrm{fail,block}}^{\mathrm{BSC}}
   & \le \exp\!\left(\log n - \frac{t}{2}x - \frac{t}{2}\varepsilon\right) \notag \\
   & = \exp\!\left(\left(1-\frac{c}{2}\right)\log n - \frac{t}{2}\varepsilon\right) \ 
    (\text{since } tx = c\log n.)
\end{align}
Now consider the set of weight-\(pn\) erasure patterns,
\[
A_m=\{\mathbf e\in\{0,1\}^n : \mathrm{wt}(\mathbf e)=pn\}.
\]
By the classical lower bound on binomial probabilities  
\cite[Ch.~10, Lem.~7]{Coding:books/MacWilliamsS77},
\begin{align}
\label{eq:A_m_lower_bound}
P(A_m)
   & = \binom{n}{pn} p^{pn} (1-p)^{(1-p)n}
    \ge \frac{1}{\sqrt{8np(1-p)}} \notag\\
    & \ge \frac{1}{\sqrt{2n}} \hspace{0.25in} \text{($4p(1-p) \le 1$ by AM-GM inequality)} \notag\\
    & = \Theta(n^{-1/2}) = \exp\!\left(-\frac12 \log n + o(1)\right).
\end{align}
Let \(B_m \subseteq A_m\) be the subset of weight-\(pn\) patterns on which MLD \emph{fails}.  
Because all weight-\(pn\) erasure patterns have identical probability under the $\mathrm{BSC}$ (i.e., they have the same \emph{type}~\cite[Ch.~2]{Csiszár_Körner_2011}),
\begin{align*}
P(A_m)& =|A_m| p^{pn} (1-p)^{(1-p)n}, \\
P(B_m) & =|B_m| p^{pn} (1-p)^{(1-p)n}.
\end{align*}
Thus the fraction of uncorrectable weight-\(pn\) patterns is
\[
\frac{|B_m|}{|A_m|}
= \frac{P(B_m)}{P(A_m)}.
\]
Since any block failure must occur on some error pattern,
\[
P(B_m)\le P_{\mathrm{fail,block}},
\]
and therefore
\[
\frac{|B_m|}{|A_m|} = \frac{P(B_m)}{P(A_m)}
    \le \frac{P_{\mathrm{fail,block}}}{P(A_m)}.
\]

Substituting the upper bound for $P_{\mathrm{fail,block}}^{\mathrm{BSC}}$ in~\eqref{eq:fail_block} and the lower bound for $P(A_m)$ in~\eqref{eq:A_m_lower_bound}, one has
\begin{align}
\label{eq:exponent_bound}
\frac{|B_m|}{|A_m|}
\le
\exp\!\left(\frac32\log n - \frac{c}{2}\log n - \frac{t}{2}\varepsilon - o(1)\right).
\end{align}
If \(c \ge 3\) and \(t(n)\to\infty\), the exponent tends to \(-\infty\), and hence
\[
\frac{|B_m|}{|A_m|} \longrightarrow 0.
\]

Thus, asymptotically, only a vanishing fraction of weight-\(pn\) error patterns are uncorrectable, proving that virtually all such patterns are correctly decoded.
\end{proof}
\textcolor{black}{Note that although the proof of Theorem~\ref{thm:error_weight} is presented for error patterns of weight $\frac{n}{2}\!\left(1 - \sqrt[2r]{\frac{c \log n}{t(n)}}\right)$, it is obvious to see that almost all error pattern of weight \emph{less} than this values is also correctable, simply by adjusting the value of $p$ in~\eqref{eq:setting_p}.}

\textcolor{black}{Theorem~\ref{thm:error_weight} can be recast as the following.
\begin{corollary}
Long RM codes can be decoded with vanishing output (block) error probability and complexity $O(nt(n))$ on a BSC with bit flip (or crossover) probability $(1-\epsilon)/2$.
\end{corollary}}
One notes that the constant \(c\) in~\eqref{eq:error_weight} and~\eqref{eq:exponent_bound} captures a trade-off between the \emph{weight} of the correctable error patterns and the \emph{rate} at which the fraction of uncorrectable patterns decays.  
Choosing a smaller value of \(c\) increases the threshold weight \(w\) in~\eqref{eq:error_weight}, allowing MLD to correct error patterns of higher weight.  
However, this comes at the cost of a slower decay in the upper bound~\eqref{eq:exponent_bound} on the fraction of uncorrectable patterns.  
Conversely, larger values of \(c\) yield faster decay of this fraction but guarantee correctability only for error patterns of smaller weight. 

On the other hand, having larger values of $t$ helps to yield faster decay on the fraction of uncorrectable patterns, and also guarantee correctability for error patterns of higher weight.

\begin{remark}
Unlike Theorem~\ref{thm:block_flip_prob}, Theorem~\ref{thm:error_weight} does not impose any condition on the growth rate of \(t(n)\). However, if \(t(n) = \omega(\log n)\) (as in the setting of Theorem~\ref{thm:block_flip_prob}), then the weight threshold
\[
    w = \frac{n}{2}\left(1 - \sqrt[2r]{\frac{c \log n}{t(n)}}\right)
\]
satisfies \(\lim_{n\to\infty} \frac{w}{n} = \frac{1}{2}\). In this regime, MLD can asymptotically correct \emph{virtually all} error patterns of weight up to half the entire blocklength \(n\). In particular, the typical error patterns corrected by MLD may have weight far exceeding \(\frac{d_{\min}}{2}\) (i.e., half the minimum distance of the underlying code)—a phenomenon that reflects the large gap between worst-case (deterministic) and typical-case (probabilistic) error correction capability.

In contrast, Theorem~\ref{thm:block_flip_prob} does not explicitly define a weight threshold for correctable error patterns. Instead, relying on the assumption of a bit flip probability \(p_f \in [0, 0.5)\), the Law of Large Numbers implies that the weight of the error patterns in a $\mathrm{BSC}(p_f)$ channel will concentrate around \(n p_f < n/2\), which is consistent with the asymptotic bound derived in Theorem~\ref{thm:error_weight}.
\end{remark}


\begin{exmp}[Numerical example]
For $n=2^{10}$, $t(n)=\log ^2 n$, $r=4$, and $c=2$:
\begin{align*}
\log  n =\log  1024 & = 10,\quad t = 100,\\
\left(\frac{c\log  n}{t}\right)^{\!1/8} & =\left(\frac{2}{10}\right)^{\!1/8}\approx 0.818.
\end{align*}
Thus
\[
w\approx 2^{10}\,(1-0.818)\approx 186.
\]
So with $t=\log ^2 n$, MLD succeeds for \emph{virtually all} error patterns up to $w\approx 186$.  
By contrast, the \textbf{worst-case deterministic availability guarantee} only ensures recovery from up to $\lfloor\frac{t-1}{2}\rfloor = 49$ errors (e.g., if each hits a different recovery set). There is a huge difference in decoding performance between the typical-case and the worst-case scenarios.
\end{exmp}

\subsection{Binary Erasure Channel ($\mathrm{BEC}$)}
In this case, the communication channel is a $\mathrm{BEC}$ where each bit is independently erased with probability $p_e$ (``data loss'') where $0 \le p_e < 1$.

For a given recovery set $R^{x, k}$ of size $r$:
\[
\Pr[\text{no erasure in }R^{x, k}] = (1-p_e)^r.
\]
Since the $t$ recovery sets are pairwise disjoint (hence \emph{independent}), the probability that
\emph{all} sets contain at least one erasure is
\[
\Pr[\text{all }R^{x, k} \text{ have $\ge1$ erasure}] = \big(1 - (1-p_e)^r\big)^t.
\]
This leads us to the following result.
\begin{proposition}\label{prop:erasure_prob_bec}
The probability of \emph{symbol/bit decoding failure} (all sets contains at least one erased symbol) over a $\mathrm{BEC}$ where bits are erased independently with probability $p_e$ is
\[
P_{\mathrm{fail,bit}}^{\mathrm{BEC}}
= \big(1 - (1-p_e)^r\big)^t.
\]
\end{proposition}
Since $0 \le 1 - (1-p_e)^r < 1$, the probability of decoding failure decays as $t$ grows. This observation leads to the following result.

\begin{theorem}[\textbf{MLD for LRCs with availability $t$ and locality $r$ over the $\mathrm{BEC}$}]
\label{thm:erasure_prob}
Consider a binary \((r,t)_a\)-LRC where $r$ is fixed and $t = t(n)$ grows with $n$. Suppose the codeword is transmitted over a $\mathrm{BEC}(p_e)$ with $\mathrm{i.i.d.}$ erasures, where $0 \le p_e < 1$ is fixed.  
If $t(n)$ grows faster than $\log n$ as $n$ tends to infinity, then the probability of block decoding failure vanishes in the asymptotic regime of $n$:
\[
    \lim\limits_{n\rightarrow \infty} P_{\mathrm{fail,block}}^{\mathrm{BEC}} = 0.
\]
\end{theorem}

\begin{proof}
The proof is similar to that of Theorem~\ref{thm:block_flip_prob}. We present it here for completeness. 

For any fixed symbol, Proposition~\ref{prop:erasure_prob_bec} gives the erasure decoding failure probability
\[
P_{\mathrm{fail,bit}}^{\mathrm{BEC}}
    = \big(1 - (1-p_e)^r\big)^t.
\]
Using the union bound, the failure probability of decoding blocks (or words) satisfies
\[
P_{\mathrm{fail,block}}^\mathrm{BEC}
    \;\le\; n \, P_{\mathrm{fail,bit}}^\mathrm{BEC}
    = n\big(1 - (1-p_e)^r\big)^t.
\]
Define
\[
x \triangleq (1-p_e)^{r} \in (0,1],
\]
so that the upper bound becomes \(n(1-x)^{t}\).  
Using the standard inequality \(\log(1-x) < -x\) for all \(x \in (0,1)\), we may write
\[
\log(1-x) = -x - \varepsilon,
\qquad \text{for some } \varepsilon > 0.
\]
Thus
\[
n(1-x)^{t}
    = \exp\!\left(\log n - tx - t\varepsilon\right).
\]

Since \(x = (1-p_e)^{r}\) and $\varepsilon$ are both positive, the exponent tends to \(-\infty\) provided
\[
\lim_{n\to \infty}\frac{t(n)}{\log n} \;\longrightarrow\; \infty.
\]
In this case we have \(n(1-x)^{t} \to 0\), and therefore  
\(P_{\mathrm{fail,block}}^{\mathrm{BEC}} \to 0\). This completes the proof.
\end{proof}



Following the same spirit of Theorem~\ref{thm:error_weight}, one obtains the following result.
\begin{theorem}
\label{thm:erasure_weight}
Let \(c \ge \dfrac{3}{2}\) be any constant.  
For a binary \((r,t)_a\)-LRC with availability \(t = t(n)\) growing with $n$, all but a vanishing fraction of erasure patterns of weight
\begin{equation}
\label{eq:erasure_weight}
w \;\le\; n\!\left(1 - \sqrt[r]{\frac{c \log n}{t}}\right)    
\end{equation}
are correctly recovered by majority-logic decoding as \(n \to \infty\). In other words, virtually all erasure patterns of weight $w$ are correctable in the asymptotic regime of $n$. \textcolor{black}{In other words, long RM codes of growing availability can be decoded with complexity $O(nt(n))$ and decoding threshold $\delta = n,\, \epsilon = \sqrt[r]{\frac{c \log n}{t(n)}}$.}
\end{theorem}

\begin{proof}
The proof follows the same structure as that of Theorem~\ref{thm:error_weight}; we therefore highlight only the essential differences. Set
\begin{align*}
p & \triangleq 1 - \sqrt[r]{\frac{c\log n}{t}},
\qquad\text{so }\ 0 \le p \le 1 \ \\
\text{ and }\
x & \triangleq (1-p)^r = \frac{c\log n}{t}.
\end{align*}
Using Union bound and Proposition~\ref{prop:erasure_prob_bec}, one has the block-failure bound
\[
P_{\mathrm{fail,block}}^{\mathrm{BEC}}
    \;\le\; n (1-x)^{t}.
\]
Using the identity \(\log(1-x) = -x - \varepsilon\) for some \(\varepsilon > 0\) when $x \in (0, 1)$,
\begin{align*}
n(1-x)^t
  & = \exp\!\bigl(\log n - tx - t\varepsilon\bigr)\\
    & = \exp\!\bigl((1-c)\log n - t\varepsilon\bigr)\qquad (\text{since } tx = c\log n.)
\end{align*}
Thus
\[
P_{\mathrm{fail,block}}^{\mathrm{BEC}}
    \le \exp\!\bigl((1-c)\log n - t\varepsilon\bigr).
\]
Let \(A_m\) denote the set of all error patterns of weight \(pn\), and let \(B_m\subseteq A_m\) be the subset on which MLD fails.  
As in the proof of Theorem~\ref{thm:error_weight}, all patterns in \(A_m\) have identical probability under the $\mathrm{BSC}$, and therefore
\[
\frac{|B_m|}{|A_m|}
= \frac{P(B_m)}{P(A_m)}
\le \frac{P_{\mathrm{fail,block}}}{P(A_m)}.
\]
Moreover, by the classical binomial lower bound~\cite[Ch.~10, Lem.~7]{Coding:books/MacWilliamsS77},
\[
P(A_m)
\ge \Theta(n^{-1/2})
= \exp\!\left(-\frac12 \log n + o(1)\right).
\]
Substituting the upper bound for $P_{\mathrm{fail,block}}^{\mathrm{BEC}}$ and the lower bound for $P(A_m)$,
\begin{align}
\label{eq:fraction}
\frac{|B_m|}{|A_m|}
& \;\le\;
\frac{
    \exp\!\bigl((1-c)\log n - t\varepsilon\bigr)
}{
    \exp\!\left(-\frac{1}{2}\log n + o(1)\right)
}\notag \\ 
& =
\exp\!\left(\frac{3}{2}\log n - c\log n - t\varepsilon - o(1)\right).
\end{align}
Because \(c \ge \frac{3}{2}\) and \(\varepsilon > 0\), the exponent tends to \(-\infty\) when $t(n)$ tends to $\infty$ with $n$, and thus
\[
\frac{|B_m|}{|A_m|}
\;\longrightarrow\; 0.
\]

Hence asymptotically, only a vanishing fraction of weight-\(pn\) erasure patterns are uncorrectable, proving that virtually all such patterns are correctly decoded.
\end{proof}

\begin{remark}
Comparing Theorems~\ref{thm:erasure_prob} and~\ref{thm:error_prob}, we observe that the BER over the $\mathrm{BEC}$ decays at \emph{twice} the exponential rate of the corresponding failure probability over the $\mathrm{BSC}$. Indeed, BER of the $\mathrm{BEC}$ behaves like \((1-k_1)^t\), whereas that of the $\mathrm{BSC}$ behaves like \((1-k_2)^{t/2}\), where
\[
k_1 = (1-p_e)^r
\qquad\text{and}\qquad
k_2 = (1-2p_f)^{2r}
\]
are fixed constants.  
Thus, majority-logic decoding is exponentially more robust to erasures than to bit flips. This two-fold difference in decay rates manifests itself in the asymptotic gap between the number of correctable erasures in Theorem~\ref{thm:erasure_weight} (on the order of \(n\)) and the number of correctable errors in Theorem~\ref{thm:error_weight} (on the order of \(n/2\)).  
This observation is consistent with the classical principle in coding theory that codes can typically \emph{correct twice as many erasures as errors}.
\end{remark}

\section{Simulations}
\label{sec:simulation}
To validate the theoretical bounds derived for MLD over the BSC, we performed simulations analyzing both BER and BLER. The simulations utilized a fixed locality parameter of $r=4$ while varying the availability $t(n)$ across different growth regimes relative to the block length $n$. \textcolor{black}{We intentionally choose a relatively large crossover probability \(p_f\) to demonstrate the robustness of MLD against high-weight error patterns when sufficient availability is provided.} The plotted numerical values are each averaged over $5 \times 10^6$ runs.

\subsection{Bit Decoding Failure Probability}
Figure~\ref{fig:bit_dec_error} compares the empirical bit decoding failure probability (BER) against the theoretical upper bounds established in Theorem~\ref{thm:error_prob} for a BSC with bit-flip probability $p_f = 0.2$. We examined three distinct regimes for the availability $t$: linear growth ($t=n/4$), polylogarithmic growth ($t=(\log_2 n)^2$), and sub-logarithmic growth ($t=\sqrt{\log_2 n}$). 

The results demonstrate a clear threshold behavior dependent on the growth rate of $t$:

\begin{itemize}
   \item \textbf{Linear and Polylogarithmic Regimes:} In cases where $t$ grows linearly or as $(\log_2 n)^2$, the empirical failure probability decays rapidly to zero as $n$ increases, confirming that MLD successfully corrects errors when the number of recovery sets is sufficiently large. Notably, the empirical performance is consistently superior to the theoretical upper bounds, though the gap narrows as $n$ grows. The tightness of the bound is directly related to the growth rate of $t(n)$; for the linear case ($t = n/4$), the theoretical bound closely approximates the empirical probability for $n \ge 2^{12}$, whereas it provides a much looser estimate in the slower-growth regimes (e.g., $t = \sqrt{\log n}$). Note that as mentioned in Remark~\ref{rm:tightness}, although being exponentially tight, the theoretical upper bound can be further refined by incorporating a multiplicative pre-factor.
    \item \textbf{Sub-logarithmic Regime:} \textcolor{black}{For the slowly growing availability \(t=\sqrt{\log_2 n}\), the empirical decoding failure probability does not appear to vanish over the simulated range. Instead, it saturates at a relatively high error floor (approximately \(0.2\) to \(0.4\)), suggesting that sub-logarithmic availability may be insufficient for reliable error correction via majority voting at practical blocklengths. \textcolor{black}{Meanwhile, the theoretical upper bound in this regime remains close to \(1\) for the considered blocklengths, but it does decrease with \(n\); its decay is much slower (and therefore less apparent in the plot) than in the linear and polylogarithmic regimes.}}
\end{itemize}
\subsection{Block Decoding Failure Probability}
Figure~\ref{fig:block_dec_error} focuses on the upper bound of the BLER, derived via the union bound on the Bit Error Rate (BER) (i.e., $P_{\mathrm{fail,block}} \le n P_{\mathrm{fail,bit}}$) for a BSC with $p_f = 0.13$. Since this BLER is derived from a union bound, the resulting values may exceed both the true BLER and even $1$, which is reflected in the vertical axis of the plot. We do not report empirical BLER values, as their computation becomes infeasible for large block lengths: verifying correct block decoding requires checking the correctness of all symbols, each of which involves an increasing number of recovery sets as $n$ grows. This figure provides a fine-grained analysis of the polylogarithmic regime by varying the exponent of the logarithmic growth, specifically $t=(\log_2 n)^{\alpha}$ for $\alpha \in \{1.8, 1.9, 2.05\}$, 
and highlights the sensitivity of block-level performance to the exponent \(\alpha\). For smaller values of \(\alpha\) (e.g., \(\alpha = 1.8\)), the BLER bound remains relatively large over the range of blocklengths considered and decreases only at very large \(n\), indicating that such growth of \(t(n)\) is insufficient for practical reliability. In contrast, slightly larger exponents (e.g., \(\alpha \ge 1.9\)) lead to a rapid decay of the failure probability, illustrating that the availability must scale with a sufficiently large power of \(\log n\) to offset the accumulation of error probability as the block length increases.

The simulations confirm that while MLD is robust, its asymptotic reliability on the BSC is strictly governed by the scaling law of the availability $t$. Linear availability provides the strongest protection, whereas polylogarithmic availability requires a sufficient growth exponent to ensure successful block decoding.

\begin{figure}
    \centering
    \includegraphics[width=1\linewidth]{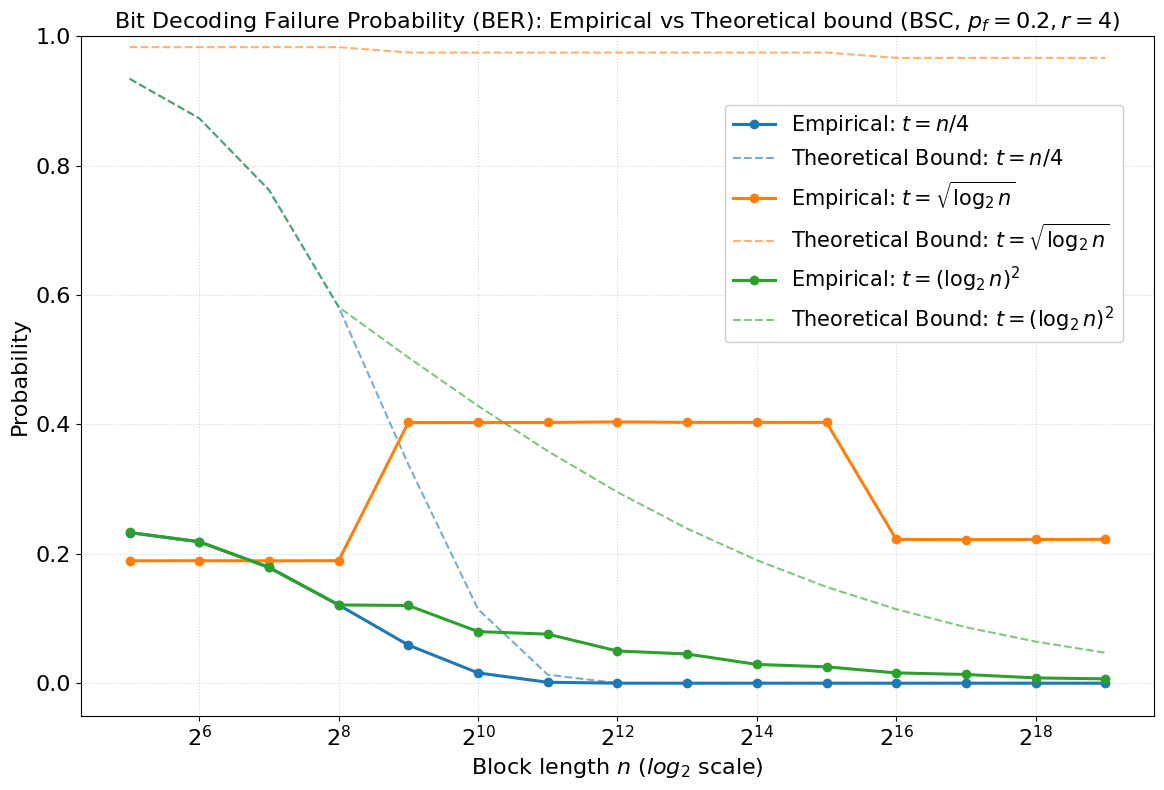}
    \caption{Bit Decoding Failure Probability: Empirical vs Theoretical ($p_f=0.2, r=4$). The plot compares linear, polylogarithmic, and sub-logarithmic availability regimes.}
    \label{fig:bit_dec_error}
\end{figure}

\begin{figure}
    \centering
    \includegraphics[width=1\linewidth]{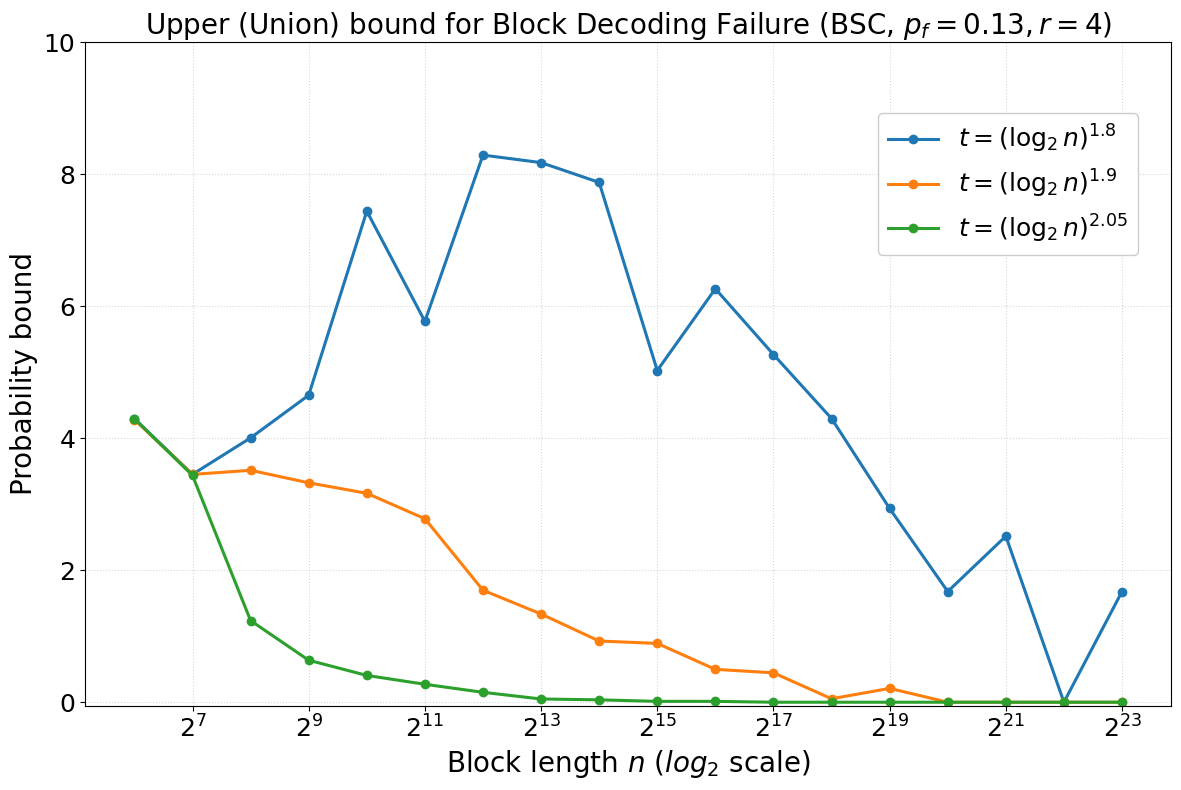}
    \caption{Upper (Union) bound for Block Decoding Failure ($p_f=0.13, r=4$). The plot demonstrates the impact of the logarithmic exponent $\alpha$ on block decoding success.}
    \label{fig:block_dec_error}
\end{figure}

\section{Conclusion}
\label{sec:conclusion}
We analyzed the performance of majority-logic decoding for binary linear locally recoverable codes over stochastic channels. Our results show that, with sufficiently large availability, LRCs achieve vanishing block error probability and can correct a linear fraction of random errors and erasures, despite much weaker worst-case guarantees. This highlights the effectiveness of majority-logic decoding in exploiting the availability structure of LRCs and motivates further study of probabilistic decoding behavior for other locally recoverable code families, especially for $q$-ary codes.

\section*{Acknowledgment}
This work was supported in part by the NSF-BSF grant FET212026.
The authors thank A.~Barg for providing them with the Russian original and its English translation of~\cite{Krichevskiy1970RM}.



\bibliography{bibliography}

@ARTICLE{LRC_large_availability,
  author={Jin, Lingfei and Kan, Haibin and Luo, Yuan and Zhang, Wenqin},
  journal={IEEE Transactions on Information Theory}, 
  title={Binary Locally Repairable Codes With Large Availability and its Application to Private Information Retrieval}, 
  year={2022},
  volume={68},
  number={4},
  pages={2203-2210},
  doi={10.1109/TIT.2022.3144034}}

@ARTICLE{Binary_LRC,
  author={Huang, Pengfei and Yaakobi, Eitan and Uchikawa, Hironori and Siegel, Paul H.},
  journal={IEEE Transactions on Information Theory}, 
  title={Binary Linear Locally Repairable Codes}, 
  year={2016},
  volume={62},
  number={11},
  pages={6268-6283},
  doi={10.1109/TIT.2016.2605119}}

@INPROCEEDINGS{Bounds_LRC,
  author={Tamo, Itzhak and Barg, Alexander},
  booktitle={2014 IEEE International Symposium on Information Theory}, 
  title={Bounds on locally recoverable codes with multiple recovering sets}, 
  year={2014},
  volume={},
  number={},
  pages={691-695},
  doi={10.1109/ISIT.2014.6874921}}

@INPROCEEDINGS{OS_MLD_memory_protection,
  author={Dupraz, Elsa and Declercq, David and Vasić, Bane},
  booktitle={2015 Information Theory and Applications Workshop (ITA)}, 
  title={Analysis of {Taylor-Kuznetsov} memory using one-step majority logic decoder}, 
  year={2015},
  volume={},
  number={},
  pages={46-53},
  doi={10.1109/ITA.2015.7308965}}

@book{Coding:books/MacWilliamsS77,
  author    = {MacWilliams, F. J. and Sloane, N. J. A.},
  title     = {The Theory of Error-Correcting
Codes},
  publisher = {Elsevier},
  year      = {1977},
  address   = {Amsterdam}
}

@book{concentration,
    author = {Boucheron, Stéphane and Lugosi, Gábor and Massart, Pascal},
    title = {Concentration Inequalities: A Nonasymptotic Theory of Independence},
    publisher = {Oxford University Press},
    year = {2013},
    month = {02},
    isbn = {9780199535255},
    doi = {10.1093/acprof:oso/9780199535255.001.0001},
    url = {https://doi.org/10.1093/acprof:oso/9780199535255.001.0001},
}

@book{Massey1963Threshold,
  author    = {J. L. Massey},
  title     = {Threshold Decoding},
  year      = {1963},
  publisher = {MIT Press},
  address   = {Cambridge, MA},
  series    = {Research Monograph},
  number    = {20}
}

@article{reed1954class,
  title={A class of multiple-error-correcting codes and the decoding scheme},
  author={Reed, Irving},
  journal={Trans.\ of the IRE Professional Group on Inform.\ Theory},
  volume={4},
  number={4},
  pages={38--49},
  year={1954},
  publisher={IEEE}
}

@book{dembo2009large,
  title={Large Deviations Techniques and Applications},
  author={Dembo, Amir and Zeitouni, Ofer},
  volume={38},
  year={2009},
  publisher={Springer Science \& Business Media},
  edition={2nd}
}

@INPROCEEDINGS{MLD_LDPC,
  author={Radhakrishnan, Rathnakumar and Sankaranarayanan, Sundararajan and Vasic, Bane},
  booktitle={2007 IEEE International Symposium on Information Theory}, 
  title={Analytical Performance of One-Step Majority Logic Decoding of Regular {LDPC} Codes}, 
  year={2007},
  volume={},
  number={},
  pages={231-235},
  doi={10.1109/ISIT.2007.4557231}}

@article{LRC_availablity_construction,
author = {Micheli, Giacomo and Pallozzi Lavorante, Vincenzo and Shukul, Abhi and Smith, Noah},
title = {Constructions of locally recoverable codes with large availability},
year = {2025},
issue_date = {Aug 2025},
publisher = {Kluwer Academic Publishers},
address = {USA},
volume = {93},
number = {8},
issn = {0925-1022},
url = {https://doi.org/10.1007/s10623-025-01624-w},
doi = {10.1007/s10623-025-01624-w},
journal = {Des. Codes Cryptography},
month = apr,
pages = {2931–2945},
numpages = {15}
}

@ARTICLE{LRC_scaling_locality,
  author={Rawat, Ankit Singh and Papailiopoulos, Dimitris S. and Dimakis, Alexandros G. and Vishwanath, Sriram},
  journal={IEEE Transactions on Information Theory}, 
  title={Locality and Availability in Distributed Storage}, 
  year={2016},
  volume={62},
  number={8},
  pages={4481-4493},
  doi={10.1109/TIT.2016.2524510}}

@ARTICLE{Capacity_LRC,
  author={Mazumdar, Arya},
  journal={IEEE Journal on Selected Areas in Information Theory}, 
  title={Capacity of Locally Recoverable Codes}, 
  year={2023},
  volume={4},
  number={},
  pages={276-285},
  doi={10.1109/JSAIT.2023.3300901}}

@book{harris2013digital,
  title     = {Digital Design and Computer Architecture},
  author    = {Harris, David Money and Harris, Sarah L.},
  edition   = {2nd},
  publisher = {Morgan Kaufmann},
  year      = {2013},
  isbn      = {978-0-12-394424-5}
}

@INPROCEEDINGS{arbitrary_locality_availability,
  author={Wang, Anyu and Zhang, Zhifang and Liu, Mulan},
  booktitle={2015 IEEE International Symposium on Information Theory (ISIT)}, 
  title={Achieving arbitrary locality and availability in binary codes}, 
  year={2015},
  volume={},
  number={},
  pages={1866-1870},
  doi={10.1109/ISIT.2015.7282779}}

@article{Codes_distributed_storage,
url = {http://dx.doi.org/10.1561/0100000115},
year = {2022},
volume = {19},
journal = {Foundations and Trends® in Communications and Information Theory},
title = {Codes for Distributed Storage},
doi = {10.1561/0100000115},
issn = {1567-2190},
number = {4},
pages = {547-813},
author = {Vinayak Ramkumar and S. B. Balaji and Birenjith Sasidharan and Myna Vajha and M. Nikhil Krishnan and P. Vijay Kumar}
}

@INPROCEEDINGS{local-error-correction,
  author={Prakash, N. and Kamath, Govinda M. and Lalitha, V. and Kumar, P. Vijay},
  booktitle={2012 IEEE International Symposium on Information Theory Proceedings}, 
  title={Optimal linear codes with a local-error-correction property}, 
  year={2012},
  volume={},
  number={},
  pages={2776-2780},
  doi={10.1109/ISIT.2012.6284028}}

@article{One_step_RM,
author={Hoang Ly and Emina Soljanin},      
title={Optimum 1-Step Majority-Logic Decoding of Binary {Reed-Muller} Codes}, 
journal = {arXiv preprint arXiv:2508.08736},      
year={2025},
      url={https://arxiv.org/abs/2508.08736} 
}

@article{Krichevskiy1970RM,
  author  = {Krichevskiy, R. E.},
  title   = {On the Number of Reed--Muller Code Correctable Errors},
  journal = {Doklady Akademii Nauk SSSR},
  volume  = {191},
  pages   = {541--547},
  year    = {1970},
  language = {Russian}
}

@book{Coding:books/PetersonW72,
  author    = {Peterson, W. W. and Weldon, Jr. E. J.},
  title     = {Error-Correcting Codes},
  publisher = {The MIT Press},
  year      = {1972},
  edition   = {2nd}
}

@ARTICLE{MLD_memory,
  author={Liu, Shih-Fu and Reviriego, Pedro and Maestro, Juan Antonio},
  journal={IEEE Trans.\ on Very Large Scale Integration (VLSI) Systems}, 
  title={Efficient Majority Logic Fault Detection With Difference-Set Codes for Memory Applications}, 
  year={2012},
  volume={20},
  number={1},
  pages={148-156},
  doi={10.1109/TVLSI.2010.2091432}}

@ARTICLE{Class_Of_MLDcodes,
  author={Lin, Shu and Markowsky, George},
  journal={IBM Journal of Research and Development}, 
  title={On a Class of One-Step Majority-Logic Decodable Cyclic Codes}, 
  year={1980},
  volume={24},
  number={1},
  pages={56-63},
  doi={10.1147/rd.241.0056}}

@article{RMdecoding:journals/tcom/BertramHH13,
  author       = {Juliane Bertram and
                  Peter Hauck and
                  Michael Huber},
  title        = {An Improved Majority-Logic Decoder Offering Massively Parallel Decoding
                  for Real-Time Control in Embedded Systems},
  journal      = {{IEEE} Trans. Commun.},
  volume       = {61},
  number       = {12},
  pages        = {4808--4815},
  year         = {2013},
  url          = {https://doi.org/10.1109/TCOMM.2013.102413.130109},
  doi          = {10.1109/TCOMM.2013.102413.130109},
  bibsource    = {dblp computer science bibliography, https://dblp.org}
}

@article{Grassmannian_code_MJ,
title = {Point-line incidence on {G}rassmannians and majority logic decoding of {G}rassmann codes},
journal = {Finite Fields and Their Applications},
volume = {73},
pages = {101843},
year = {2021},
issn = {1071-5797},
doi = {https://doi.org/10.1016/j.ffa.2021.101843},
url = {https://www.sciencedirect.com/science/article/pii/S107157972100037X},
author = {Peter Beelen and Prasant Singh}
}

@ARTICLE{RMdecoding:journals/tit/Chen71,
  author={Chin-Long Chen},
  journal={IEEE Trans.\ on Inform.\ Theory}, 
  title={On majority-logic decoding of finite geometry codes}, 
  year={1971},
  volume={17},
  number={3},
  pages={332-336},
  doi={10.1109/TIT.1971.1054629}}

@ARTICLE{Gopalan,
  author={Gopalan, Parikshit and Huang, Cheng and Simitci, Huseyin and Yekhanin, Sergey},
  journal={IEEE Transactions on Information Theory}, 
  title={On the Locality of Codeword Symbols}, 
  year={2012},
  volume={58},
  number={11},
  pages={6925-6934},
  doi={10.1109/TIT.2012.2208937}}

@ARTICLE{multiple_erasures,
  author={Wang, Anyu and Zhang, Zhifang},
  journal={IEEE Transactions on Information Theory}, 
  title={Repair Locality With Multiple Erasure Tolerance}, 
  year={2014},
  volume={60},
  number={11},
  pages={6979-6987},
  doi={10.1109/TIT.2014.2351404}}

@ARTICLE{soft_decision_Dumer,
  author={Dumer, I. and Krichevskiy, R.},
  journal={IEEE Transactions on Information Theory}, 
  title={Soft-decision majority decoding of Reed-Muller codes}, 
  year={2000},
  volume={46},
  number={1},
  pages={258-264},
  doi={10.1109/18.817523}}

@INPROCEEDINGS{SRR_Design:lySV2025,
  author={Ly, Hoang and Soljanin, Emina},
  booktitle={61th Annual Allerton Conference on Communication, Control, and Computing}, 
  title={{Maximal Achievable Service Rates of Codes and Connections to Combinatorial Designs}}, 
  year={2025},
  volume={},
  number={},
  pages={},
  doi={},
  journal={arXiv preprint arXiv:2506.16983},
  url={https://arxiv.org/abs/2506.16983}
}

@article{cosmic_rays_aircraft_2025,
  author       = {Chris Baraniuk},
  title        = {Bit flips: How cosmic rays grounded a fleet of aircraft},
  journal      = {BBC Future},
  year         = {2025},
  month        = dec,
  day          = {1},
  url          = {https://www.bbc.com/future/article/20251201-how-cosmic-rays-grounded-thousands-of-aircraft},
  note         = {Accessed: 2026-01-13}
}

@ARTICLE{PIR_LRC_equivalent,
  author={Kadhe, Swanand and Heidarzadeh, Anoosheh and Sprintson, Alex and Koyluoglu, O. Ozan},
  journal={IEEE Journal on Selected Areas in Information Theory}, 
  title={Single-Server Private Information Retrieval Schemes are Equivalent to Locally Recoverable Coding Schemes}, 
  year={2021},
  volume={2},
  number={1},
  pages={391-402},
  doi={10.1109/JSAIT.2021.3053579}}

@ARTICLE{5_seq_Golomb,
  author={Choi, Hyojeong and Song, Hong-Yeop},
  journal={IEEE Transactions on Information Theory}, 
  title={Optimal 5-Seq LRCs With Availability From Golomb Rulers}, 
  year={2025},
  volume={71},
  number={3},
  pages={1689-1699},
  doi={10.1109/TIT.2025.3525668}}

@ARTICLE{RM_recursive_Dumer,
  author={Dumer, I.},
  journal={IEEE Transactions on Information Theory}, 
  title={Recursive decoding and its performance for low-rate Reed-Muller codes}, 
  year={2004},
  volume={50},
  number={5},
  pages={811-823},
  doi={10.1109/TIT.2004.826632}}

@book{Csiszár_Körner_2011, place={Cambridge}, edition={2}, title={Information Theory: Coding Theorems for Discrete Memoryless Systems}, publisher={Cambridge University Press}, author={Csiszár, Imre and Körner, János}, year={2011}}

@articleInfo{product_code,
title = {An approach to the performance of SPC product codes on the erasure channel},
journal = {Advances in Mathematics of Communications},
volume = {10},
number = {1},
pages = {11-28},
year = {2016},
issn = {1930-5346},
doi = {10.3934/amc.2016.10.11},
url = {https://www.aimsciences.org/article/id/0d6ac0f5-92c6-46d7-b68d-214f19e46307},
author = {Sara D.  Cardell and Joan-Josep  Climent}
}

@ARTICLE{Tornado_codes,
  author={Luby, M.G. and Mitzenmacher, M. and Shokrollahi, M.A. and Spielman, D.A.},
  journal={IEEE Transactions on Information Theory}, 
  title={Efficient erasure correcting codes}, 
  year={2001},
  volume={47},
  number={2},
  pages={569-584},
  doi={10.1109/18.910575}}

@INPROCEEDINGS{LT_code,
  author={Luby, M.},
  booktitle={The 43rd Annual IEEE Symposium on Foundations of Computer Science, 2002. Proceedings.}, 
  title={LT codes}, 
  year={2002},
  volume={},
  number={},
  pages={271-280},
  doi={10.1109/SFCS.2002.1181950}}

@ARTICLE{Raptor_code,
  author={Shokrollahi, A.},
  journal={IEEE Transactions on Information Theory}, 
  title={Raptor codes}, 
  year={2006},
  volume={52},
  number={6},
  pages={2551-2567},
  doi={10.1109/TIT.2006.874390}}

@INPROCEEDINGS{LRC_download_time_analysis,
  author={Kadhe, Swanand and Soljanin, Emina and Sprintson, Alex},
  booktitle={2015 IEEE International Symposium on Information Theory (ISIT)}, 
  title={Analyzing the download time of availability codes}, 
  year={2015},
  volume={},
  number={},
  pages={1467-1471},
  doi={10.1109/ISIT.2015.7282699}}

@ARTICLE{interleaver,
  author={Shi, Y.Q. and Xi Min Zhang and Zhi-Cheng Ni and Ansari, N.},
  journal={IEEE Circuits and Systems Magazine}, 
  title={Interleaving for combating bursts of errors}, 
  year={2004},
  volume={4},
  number={1},
  pages={29-42},
  doi={10.1109/MCAS.2004.1286985}}

@ARTICLE{irregular_LDPC,
  author={Luby, M.G. and Mitzenmacher, M. and Shokrollahi, M.A. and Spielman, D.A.},
  journal={IEEE Transactions on Information Theory}, 
  title={Improved low-density parity-check codes using irregular graphs}, 
  year={2001},
  volume={47},
  number={2},
  pages={585-598},
  doi={10.1109/18.910576}}
\bibliographystyle{IEEEtran}

\end{document}